\documentclass[11pt,reqno]{amsart}
\usepackage[utf8]{inputenc}
\usepackage{braket,amsmath,amsfonts,amsthm,amsxtra,amssymb,url,graphicx,url,xcolor,hyperref,subcaption,appendix,comment,bbm}

\newcommand{\mel}{\end{eqnarray*}}
\def\fr{\begin{align*}}

\newcommand{\kl}{\pl \le \pl}
\newcommand{\gl}{\pl \ge \pl}

\newcommand{\lel}{\pl = \pl}

\newcommand{\nz}{{\mathbb N}}
\newcommand{\nen}{n \in \nz}

\newcommand{\rz}{{\mathbb R}}
\newcommand{\zz}{{\mathbb Z}}

\newcommand{\ten}{\otimes}

\newcommand{\pl}{\hspace{.1cm}}
\newcommand{\pll}{\hspace{.3cm}}
\newcommand{\pla}{\hspace{1.5cm}}

\newcommand{\al}{\alpha}
\newcommand{\si}{\sigma}

\newcommand{\la}{\lambda}
\newcommand{\eps}{\varepsilon}

\renewcommand{\phi}{\varphi}

\newcommand{\F}{{\mathcal F}}
\newcommand{\E}{{\mathcal E}}

\newcommand{\B}{{\mathcal B}}

\newcommand{\M}{{\mathcal M}}

\newcommand{\N}{{\mathcal N}}

\newcommand{\norm}[2]{\parallel \! #1 \! \parallel_{#2}}

\newtheorem{lemma}{Lemma}[section]
\newtheorem{prop}[lemma]{Proposition}
\newtheorem{theorem}[lemma]{Theorem}
\newtheorem{cor}[lemma]{Corollary}

\newtheorem{rem}[lemma]{Remark}
\newtheorem{definition}[lemma]{Definition}

\newcommand{\re}{\begin{rem}\rm}
\newcommand{\mar}{\end{rem}}

\newcommand{\qd}{\end{proof}\vspace{0.5ex}}
\newcommand{\prf}{\begin{proof}[\bf Proof:]}

\newcommand{\xspace}{\hbox{\kern-2.5pt}}

\newtheorem*{theorem*}{Theorem}

\allowdisplaybreaks

\usepackage{color}

\definecolor{LightGray}{rgb}{0.94,0.94,0.94}
\definecolor{VeryLightBlue}{rgb}{0.9,0.9,1}
\definecolor{LightBlue}{rgb}{0.8,0.8,1}
\definecolor{DarkBlue}{rgb}{0,0,0.6}
\definecolor{LightGreen}{rgb}{0.88,1,0.88}
\definecolor{MidGreen}{rgb}{0.6,1,0.6}
\definecolor{DarkGreen}{rgb}{0,0.6,0}
\definecolor{DarkGrreen}{rgb}{0,0.8,0}

\definecolor{VeryLightYellow}{rgb}{1,1,0.9}
\definecolor{LightYellow}{rgb}{1,1,0.6}
\definecolor{MidYellow}{rgb}{1,1,0.5}
\definecolor{DarkYellow}{rgb}{0.8,1,0.3}
\definecolor{VeryLightRed}{rgb}{1,0.9,0.9}
\definecolor{LightRed}{rgb}{1,0.8,0.8}
\definecolor{DarkRed}{rgb}{0.8,0.2,0}
\definecolor{DarkRedb}{rgb}{0.6,0.2,0}
\definecolor{DarkLila}{rgb}{0.8,0,1}
\definecolor{Beige}{rgb}{0.96,0.96,0.86}
\definecolor{Gold}{rgb}{1.,0.84,0.}
\definecolor{Goldb}{rgb}{0.7,0.3,0.5}
\definecolor{MyYellow}{rgb}{1.,0.84,0.8}

\usepackage{verbatim}
\usepackage{amsmath}
\usepackage{amsthm}
\usepackage{graphicx}
\usepackage{amsfonts}
\usepackage{latexsym}
\usepackage{color}
\usepackage{braket}

\DeclareMathOperator{\tr}{tr}

	\newcommand{\RR}{\mathbb{R}}
	
	\newcommand{\TT}{\mathbb{T}}
	\newcommand{\BB}{\mathbb{B}}
	\newcommand{\NN}{\mathbb{N}}

	\newcommand{\id}{\mathbf{1}}

\begin{document}

\title{Strong converse rate for asymptotic hypothesis testing in type III}

%

\author{Marius Junge}
\address{Marius Junge, University of Illinois, Urbana, IL 61801, USA} 

\author{Nicholas LaRacuente}

\address{Nicholas Laracuente, Indiana University, Bloomington, IN 47405, USA} 
\email{nlaracu@iu.edu}

\thanks{MJ is partially supported by NSF DMS 2247114.}


\footnotetext{Corresponding author: Nicholas LaRacuente, Indiana University, nlaracu@iu.edu}

\begin{abstract}
    We extend from the hyperfinite setting to general von Neumann algebras Mosonyi and Ogawa's (2015) and Mosonyi and Hiai's (2023) results showing the operational interpretation of sandwiched relative R\'enyi entropy in the strong converse of hypothesis testing. The specific task is to distinguish between two quantum states given many copies. We use a reduction method of Haagerup, Junge, and Xu (2010) to approximate relative entropy inequalities in an arbitrary von Neumann algebra by those in finite von Neumann algebras. Within these finite von Neumann algebras, it is possible to approximate densities via finite spectrum operators, after which the quantum method of types reduces them to effectively commuting subalgebras. Generalizing beyond the hyperfinite setting shows that the operational meaning of sandwiched R\'enyi entropy is not restricted to the matrices but is a more fundamental property of quantum information. Furthermore, applicability in general von Neumann algebras opens potential new connections to random matrix theory and the quantum information theory of fundamental physics.
\end{abstract}

\maketitle

\section{Introduction}
Connections between hypothesis testing and entropy have presented since the beginnings of information theory as a field \cite{chernoff_measure_1952, blahut_hypothesis_1974, han_strong_1989, csiszar_generalized_1995}. The quantum version has also attracted much attention \cite{hiai_proper_1991, hayashi_optimal_2002, hayashi_error_2007, audenaert_discriminating_2007, jaksic_quantum_2012, buscemi_information-theoretic_2019}. Recall the hypothesis testing problem for state discrimination: given some number of copies of an unknown state $\omega$ with the promise that $\omega \in \{\rho, \eta\}$, one may construct tests, two-outcome measurements to distinguish $\omega$. 

We may interpret a state $\rho$ (within finite dimension, a density) as a functional on observables - within finite dimension, $\rho(T) = \tr(\rho T)$ for a state $\rho$ and observable $T$. However, the expression $\rho(T)$ remains meaningful even in infinite-dimensional von Neumann algebras without a finite trace. More specifically , we consider a setting in which one is given many copies of a quantum state. The set of $n$-copy tests in a von Neumann algebra $\N$ is denoted by $\TT_n := \{T_n | 0 \leq T_n \leq 1\} \subset \N^{\otimes n}$. As a two-outcome POVM, the outcome probabilities are given on normal states by $\omega^{\otimes n} \mapsto (\omega^{\otimes n}(T_n), 1 - \omega^{\otimes n}(T_n))$. Given $n$ copies of an unknown state that is either $\rho$ or $\eta$, the error probability of the first kind or type I error probability is given by
\[ \alpha_n(T_n) := \rho^{\otimes n} (\id^{\otimes n} - T_n) \pl, \]
representing the probability that the state is mistakenly identified as $\rho$ when it is actually $\eta$. The error probability of the second kind or type II error probability is given by
\[ \beta_n(T_n) := \eta^{\otimes n} (T_n) \pl, \]
representing the probability that the state is mistakenly identified as $\eta$ when it is actually $\rho$. 
The $n$th minimum type I error probability of Hoeffding type is defined \cite{hiai_quantum_2023} as
\[ \alpha_{e^{-n r}}^*(\rho^{\otimes n} \| \eta^{\otimes n}) := \min_{0 \leq T_n \leq 1} \big \{ \rho^{\otimes n}(\id - T_n) : \eta^{\otimes n}(T_n) \leq e^{- n r} \big \} \pl, \]
where the minimum is taken in all tests in $\M^{\bar{\otimes} n}$ with $\eta^{\otimes n}(T_n) \leq e^{- n r}$. Correspondingly,
\begin{equation} \label{eq:1minusalpha}
    1 - \alpha_{e^{-n r}}^*(\rho^{\otimes n} \| \eta^{\otimes n}) = \max_{0 \leq T_n \leq 1} \big \{ \rho^{\otimes n}(T_n) : \eta^{\otimes n}(T_n) \leq e^{- n r} \big \}
\end{equation}
is the $n$th maximum type I success probability. To interpret the Hoeffding error/success probabilities, note that one can always avoid the possibility of one kind of error, trivally by always guessing one state or the other. With many copies, one may soften this restriction to require merely that the probability of a type II error decreases exponentially with the number of copies. In the many-copy limit or Shannon regime, the probability of the type I error may converge to zero while the probability of a type II error approaches a fixed value. The parameter $r$ in \eqref{eq:1minusalpha} bounds the rate at which errors of the second kind are required to decrease, while the value of the left-hand side given the best possible type I error probability.. A closely related quantity is
\begin{align} \label{eq:Br}
    B_r(\rho \| \eta) := \inf \Big \{ R | \exists \{T_n\}_{n=1}^\infty, 0 \leq T_n \leq 1
    , \limsup_{n \rightarrow \infty} \eta^{\otimes n} (T_n) \leq e^{- r n}, 
      \liminf_{n \rightarrow \infty} \rho^{\otimes n} (T_n) \geq e^{-R n} \} \pl,
\end{align}
which was defined in \cite[Equation (43)]{mosonyi_quantum_2015} and denoted ``$B_e^*(r)$'' therein. $B_r(\rho \| \eta)$ is described as a strong converse exponent - it is the rate by which the type I error probability converges exponentially to 1 when the type II error probability is forced to be too small. Within the finite-dimensional setting, it was observed that
\begin{equation} \label{eq:bralpha}
    \begin{split}
    B_r(\rho\|\eta) 
    & = - \lim_{n \rightarrow \infty} \frac{1}{n} \log \Big \{ 1 - \alpha_{e^{-nr}}^* (\rho^{\otimes n} \| \eta^{\otimes n}) \Big \} \pl.
    \end{split}
\end{equation}
In general von Neumann algebras, it is no longer trivial that the limit exists. We confirm this inequality generalizes as Proposition \ref{prop:Br-alpha}.

The sandwiched R\'enyi relative entropy \cite{wilde_strong_2014, muller-lennert_quantum_2013} is given by
\[ D_\alpha^*(\rho \| \eta) = \| \rho^* \rho \|_{L^{1/2}_\alpha(\eta)} \pl, \]
where the norm is in the Kosaki $L_\alpha^{1/2}$ space weighted by $\eta$ \cite{kosaki_applications_1984, gao_capacity_2018}. The Hoeffding anti-divergence with respect to the sandwiched R\'enyi relative entropy is given for a pair of states $\rho, \eta$ by
\begin{equation} \label{eq:Hr}
    H_r^*(\rho \| \eta) = \sup_{\alpha > 1} \frac{\alpha - 1}{\alpha} \Big \{ r - D_\alpha^*(\rho \| \eta) \Big \}  \pl.
\end{equation}
The main result of this paper is then:
\begin{theorem} \label{thm:main}
    For any $r \in \RR^+$ and states $\rho, \eta \in \M_*$ on any von Neumann algebra $\M$ such that $a \rho \leq \eta$ for some $a > 0$,
    \[B_r(\rho\|\eta) = H_r^*(\rho \| \eta) \pl. \]
\end{theorem}
The finite-dimensional case of Theorem \ref{thm:main} was shown by Mosonyi and Ogawa as \cite[Theorem 4.10]{mosonyi_quantum_2015}. Mosonyi presented an infinite-dimensional generalization in \cite{mosonyi_strong_2023} for trace class or density operators. Hyperfinite von Neumann algebras are limits of finite-dimensional matrix algebras. The hyperfinite case of \ref{thm:main} was shown by Mosonyi and Hiai as \cite[Theorem 3.9]{hiai_quantum_2023}. These follow classical results of Han and Kobayashi \cite{han_strong_1989} and Csisz\'ar \cite{csiszar_generalized_1995}. A later work \cite{cheng_strong_2024} showed an improved, simplified version of one direction of the inequality (in this work, \ref{lem:converse}), that $B_r(\rho \| \eta) \geq H_r^*(\rho \| \eta)$. The majority of the current work is on showing the other direction. It had remained open to generalize the original \cite[Theorem 4.10]{mosonyi_quantum_2015} to non-hyperfinite type III.

A primary implication of Theorem \ref{thm:main} is an operational interpretation of the sandwiched relative R\'enyi entropies beyond the setting of matrix algebras or their limits. This is highlighted via Corollary \ref{cor:cutoff}, showing equivalence to the generalized cutoff rates as in \cite{csiszar_generalized_1995}. A motivating issue for earlier works was that despite their good mathematical properties, the sandwiched R\'enyi entropies initially lacked a correspondence to an operational task in information theory. Hence a major theme of \cite{mosonyi_quantum_2015} and of follow-up works \cite{mosonyi_strong_2023, hiai_quantum_2023} is that via strong converse rates and cutoff rates in hypothesis testing, the sandwiched relative R\'enyi entropies are indeed operational. A major application of quantum information theory today is to quantum field theory and related areas of fundamental physics \cite{witten_aps_2018}. Though many of these settings are hyperfinite, non-hyperfiniteness appears to be a fundamental distinction and meaningful in the theory of computation and quantum correlations \cite{ji_mip*re_2022}. An wider, analogous question is if and where the hyperfinite vs. non-hyperfinite distinction appears in information theory. In this particular case, we show that it does not.

A broader intuition from our methods combines the continuity properties of $H_r^*$ with the monotonicity properties of $B_r$ to obtain desired results. The Haagerup reduction used herein approximates general von Neumann algebras not by finite-dimensional matrices, but by infinite-dimensional algebras that still have a finite trace. Although $B_r$ is not obviously continuous, we note as Lemma \ref{lem:bdp} that it obeys a sort of reverse data processing. Therefore, successive approximations in smaller algebras only increase this quantity.

\section{Background}
Here we briefly recall the notation of von Neumann algebras and Kosaki spaces as well as the relative entropy and generalized divergences. For more detail, we refer the reader to Mosonyi and Hiai's recent work in hyperfinite algebras \cite{hiai_quantum_2023}. One may also see earlier references on von Neumann algebra types \cite{araki_classification_1968} and their physical interpretations \cite{witten_aps_2018}. In particular, we recall that a von Neumann algbera of type $II$ or $III$ is infinite-dimensional and not necessarily approximated by finite-dimensional matrix algebras \cite{ji_mip*re_2022}. An algebra of type $II_1$ does have a finite trace, despite not necessarily being a limit of matrix algebras. An algebra of type $III$ lacks any notion of a finite trace, although some quantum information notions such as relative entropies are still well-defined and well-behaved.

We first recall the crossed product:
\begin{definition}
Let $G$ be a locally compact abelian group equipped with Haar measure $dg$, and
$\hat G$ its dual group equipped with Haar measure $d\hat g$. We choose
$dg$ and $d\hat g$ so that the Fourier inversion theorem holds. Let
$\alpha$ be a continuous automorphic representation of $G$ on
$\M$. The crossed product $\M \rtimes_\alpha G$ is the von Neumann
algebra on $L_2(G, H)$ generated by the operators $\pi_\alpha(x)$,
$x\in \M$ and $\lambda(g)$, $g\in G$, which are defined by
 $$\big(\pi_{\alpha} (x) \xi \big) (h)=\alpha^{-1}_{h}(x)
 \xi (h),\quad
 \big(\lambda(g) \xi\big) (h) = \xi (h-g), \quad  \xi\in L_2(G, H),\;
 h\in G.$$
\end{definition}
In particular, $\M \rtimes_\alpha \RR$ is of type $II_\infty$ even if $\M$ was of type $III$. As it has a semifinite trace, $\M \rtimes_\alpha \RR$ supports the Haagerup $L_p$ spaces as analogues of the finite-dimensional Schatten classes. For a faithful state $\eta$, by $L_p^{1/2}(\eta)$ we denote the Kosaki $L_p$, the completion of $\N$ with norm
\[ \| \eta^{1/2} X \eta^{1/2} \|_{L_p^{1/2}(\eta)} = \| \eta^{ 1/2p} X \eta^{1/2p} \|_{L_p} \pl,  \pl X\in \N \]
where $L_p$ is the associated Haagerup $L_p$ space \cite{kosaki_applications_1984, haagerup_reduction_2010} (or $L_p(N,tr)$ if $\N$ admits a normal faithful trace $tr$). In von Neumann algebras with finite trace, we may take the usual trace and Schatten norms. In von Neumann algebras that lack even a semifinite trace, the Kosaki $L_p$ and Haagerup norms are still valid and allow us to extend the definitions of relative entropies and related quantities \cite{gao_capacity_2018, junge_multivariate_2022}. Following \cite{furuya_monotonic_2023}, we recall a generalized, parameterized, pre-logarithm $f$-divergence
\begin{equation}
    Q^f_r(\rho \| \eta) := \| \eta^{1/2} f(\Delta_{\rho | \eta}^{1/r})^{1/2*} f(\Delta_{\rho | \eta}^{1/r})^{1/2} \eta^{1/2}  \|_{L^{1/2}_r(\eta)} \pl. 
\end{equation}
The corresponding entropy expression is
\[ D^f_r(\rho \| \eta) := - 2 r \ln Q^f_r(\rho \| \eta) \pl. \]
In finite dimension, the density matrix is naturally and trivially identified with a corresponding quantum state. In a semifinite von Neumann algebra $\M$, one may also associate via an intervible mapping a state $\rho \in \M_*$ with a density operator $d_\rho$ for which $\tr(d_\rho X) = \rho(X)$ for all $X$. In semifinite von Neumann algebras, recall the function
 \[ g_{\rho,\eta}(z) := d_{\rho}^{z} d_{\eta}^{-z} \pl, \]
where $d_\rho$ and $d_\eta$ are respectively the densities corresponding to states $\rho$ and $\eta$. 
The generalized $f$-divergence may equivalently be expressed as
\begin{equation} \label{eq:semifinitefdiverge}
    Q^f_r(\rho \| \eta) = \| f(g_{\rho | \eta}^{1/r})^{1/2*} f(g_{\rho | \eta}^{1/r})^{1/2}  \|_{L^{1/2}_r(\eta)} \pl.
\end{equation}
Since the trace is necessarily infinite on all elements of a type III von Neumann algebra, the usual notion of Shannon or von Neumann entropy is never finite. The quantum relative entropy and generalized $f$-divergences, however, naturally generalize the notion of entropy to general von Neumann algebras.

Within $f$-divergences, our focus is on divergences of the form $Q_{x \mapsto x^{q}, r}$. Cases of interest include $q = r = p$, yielding the sandwiched relative R\'enyi $p$-entropies, and $r = 1, q = 1/p$, yielding the Petz-R\'enyi relative $p$-entropies. In either of these cases, $g_{\rho | \eta}^{q / 2 r *}$ is a holomorphic function of $p \in [1, \infty)$ assuming Equation \eqref{eq:ordassump}. Hence $g_{\rho | \eta}^{q / 2 r *} g_{\rho | \eta}^{q / 2 r}$ is holomorphic in $p$. This work primarily uses the sandwiched relative entropies.
\begin{rem} \label{rem:sandwich}
Although the modular operator is not assured to be holomorphic in type $III$, the sandwiched R\'enyi entropy $D_\alpha^*$ arises from setting
\begin{equation}
    Q^*_{p} =  \| \eta^{1/2} \Delta_{\rho | \eta}^{1/2*} \Delta_{\rho | \eta}^{1/2} \eta^{1/2} \|_{L^{1/2}_p(\eta)} \pl.
\end{equation}
Here we may interpret $\Delta_{\rho | \eta}^{1/2} \eta^{1/2} = \Delta_{\rho | \eta}^{1/2} (\eta^{1/2})$ and $\Delta_{\rho | \eta}^{1/2*}$ as applying to the $\eta^{1/2}$ on its left side. The sandwiched relative R\'enyi entropy is known to obey a data processing inequality \cite{jencova_renyi_2018}. Furthermore, both $Q_\alpha^*(\rho \| \eta)$ and $D_\alpha^*(\rho \| \eta)$ are monotonically non-decreasing in $\alpha$ for every pair of states $\rho$ and $\eta$. When $\rho \leq C \eta$, there exists a $\nu$ for which $\rho^{1/2} = \nu \eta^{1/2}$. One may thereby also write the sandwiched relative entropy as $\| \eta^{1/2} \nu^* \nu \eta^{1/2}\|_{L^{1/2}_p(\eta)}$. In these expressions, the $p$-dependence is relegated to the norm weighting, which is analaytic for $p \in (0,\infty)$.
\end{rem}
The generalized Hoeffding anti-divergence is defined for $r \in \RR$ by
\begin{equation} \label{eq:hoeffding1}
    H_r^f (\rho \| \eta) := \sup_{\alpha > 1} \frac{\alpha - 1}{\alpha} \Big \{ r - D_\alpha^f(\rho \| \eta) \Big \} \pl.
\end{equation}
\begin{rem} \label{rem:hoeffdingequiv}
    The Hoeffding anti-divergence with respect to the Sandwiched relative entropy, which we denote $H^*_r$, is given equivalently by
    \[ H^*_r(\rho \| \eta) = \sup_{\alpha > 1} \frac{\alpha-1}{\alpha} r - \ln Q^*_\alpha (\rho \| \eta) \pl.\]
\end{rem}
\noindent A hyperfinite von Neumann algebra contains an ascending sequence of finite dimensional subalgebras with union that is dense in the original algebra. While hyperfinite factors have many important applications in operator algebras, mathematical physics, and quantum field theory \cite{araki_classification_1968, witten_aps_2018}, evidence has emerged that hyperfinite von Neumann algebras cannot always approximate general von Neumann algebras \cite{ji_mip*re_2022}. We recall a Theorem of Mosonyi and Hiai \cite[Theorem 3.9]{hiai_quantum_2023} showing the desired result in hyperfinite von Neumann algebras:
\begin{theorem}[Mosonyi-Hiai] \label{thm:origres}
    Assume that $\M$ is an injective (hyperfinite) von Neumann algebra. Let $\rho, \eta \in \M_*^+$ be states such that $D_{\alpha_0}^*(\rho \| \eta) < +\infty$ for some $\alpha_0 > 1$. Then
    \[ \lim_{n \rightarrow \infty} - \frac{1}{n} \log \Big \{ 1 - \alpha_{e^{-nr}}^* (\rho^{\otimes n} \| \eta^{\otimes n}) \Big \} = H_r^*(\rho \| \eta) \pl. \]
\end{theorem}
\noindent This connection to the Holevo type I error probability yields an operational interpretation of the sandwiched R\'enyi entropies as strong converse rates in hypothesis testing. The main goal of the current manuscript is to remove the requirement of injectivity or hyperfiniteness.

As a technical tool, recall Stein's interpolation theorem on compatible Banach spaces:
\begin{theorem}[Stein's Interpolation, \cite{bergh_interpolation_1976}]\label{thm:stein}
Let $(A_0,A_1)$ and $(B_0,B_1)$ be two couples of Banach spaces that are each compatible. Let $\{T_z| z\in S\}\subset \BB(A_0+A_1, B_0+B_1)$ be a bounded analytic family of maps such that
\[\{T_{it}|\pl t\in \mathbb{R}\}\subset \BB(A_0,B_0)\pl ,\pl \{T_{1+it}|\pl t\in \mathbb{R}\}\subset  \BB(A_1,B_1)\pl.\]
Suppose $\Lambda_0=\sup_t{\norm{T_{it}}{\B(A_0,B_0)}}$ and  $\Lambda_1=\sup_t{\norm{T_{1+it}}{\B(A_1,B_1)}}$ are both finite, then for $0< \theta < 1$, $T_\theta$ is a bounded linear map from $(A_0,A_1)_\theta$ to $(B_0,B_1)_\theta$ and
\[\norm{T_\theta}{\BB((A_0,A_1)_\theta ,(B_0,B_1)_\theta)}\le \Lambda_0^{1-\theta}\Lambda_1^{\theta}\pl .\]
\end{theorem}
Kosaki $L_p$ space innately support interpolation, and were in fact originally defined as interpolation spaces \cite{kosaki_applications_1984}. Another important tool is the Haagerup reduction method \cite{haagerup_reduction_2010}. Sometimes, $\M \rtimes_\alpha \RR$ might be denoted as $\M \rtimes \RR$, suppressing the explicit denotation of the automorphism group. Alternatively, one may consider the crossed product with a discrete group. Here we recall the Haggerup reduction:
\begin{theorem}[\cite{haagerup_reduction_2010}, Theorem 2.1] \label{thm:haagerup}
Setting $G = \cup_n 2^{-n} \mathbb{Z} \subset \mathbb{R}$, there exists a sequence of von Neumann algebras and conditional expectations $(\E_k)_{n=1}^\infty : \M_k \rightarrow \M \rtimes G$ for which
    \begin{enumerate}
    \item[i)] $\E$ and $\tilde{\eta}=\eta\circ \E$ are faithful.
    \item[ii)] There exists an increasing family of subalgebras $\M_k$ and normal conditional expectations $\E_k: \M \rtimes G \to \M_k$ such that $\tilde{\eta}F_k=\tilde{\eta}$;
    \item[iii)] $\lim_k \|F_k(\psi)-\psi\|_{(\M \rtimes G)_*}=0$ for every normal state $\psi\in (\M \rtimes G)_*$;
    \item[iv)] For every $k$ there exists a normal faithful trace trace $\tau_k$ and a density $d_k \in L(G)$ such that for every $x \in \M_k$, $\tilde{\eta}(x) = \tau_k(d_k x)$, and $a_k\le d_k\le a_k^{-1}$ for scalars $a_k \in \RR^+$.
    Hence $\tilde{M}_k$ is of type $II_1$.
    \end{enumerate}
\end{theorem}
Though not all von Neumann algebras admit hyperfinite or finite-dimensional approximations, the Haagerup reduction shows that all von Neumann algebras admit finite approximations. If a quantity is compatible with the crossed product and Haagerup reduction, then estimates derived using finite trace transfer to general von Neumann algebras.

\section{Proof of Main Results}
The proof of Theorem \ref{thm:main} follows successive, arbitrarily precise approximations. First, using the Haagerup reduction, we approximate relative entropy in arbitrary von Neumann algebras by that in finite von Neumann algebras. Within finite von Neumann algebras, there is an easy identification between densities and states. We may then approximate each density by one with finite spectrum. Finite spectrum enables the method of types as described in \cite{hayashi_optimal_2002}, which exploits the fact that for many copies of the same states, though the dimension grows exponentially in copy-number, the number of distinct eigenvalues grows only polynomially. Ultimately, the method of types yields a projection to a commuting von Neumann algebra that approximately preserves the distinguishability of states. Because commutative von Neumann algebars are automatically hyperfinite, we may then apply results of \cite{hiai_quantum_2023} to obtain a sequence of tests achieving the desired values. Since these tests apply to states in the original von Neumann algebra, the same value is achieved there. The converse follows almost immediately from known results for relative entropy. Below we explain this derivation in detail.

We often use a semidefinite order comparison assumption between two states $\rho$ and $\eta$:
\begin{equation} \label{eq:ordassump}
    a \rho \leq \eta
\end{equation}
for some $a > 0$. This assumption guarntees that $D^*_\alpha(\rho \| \eta)$ and $D^*_\alpha(\eta \| \rho)$ have values within the interval $[0, a^{-1}]$ for all $\alpha \in [1, \infty]$. 

\subsection{Approximation via Finite Algebras} \label{sec:approx-finite-alg}
\begin{lemma} \label{lem:haagerupconv}
    Let $\eta$ be a state and $\E_n : \N \rightarrow \N_n$ a sequence of $\eta$-preserving conditional expectations for an increasing sequence of algebras $(\N_n)_{n=1}^\infty$ such that for every state $\omega \in \N_*$, $\lim_{n} \E_n(\omega)$ converges to $\omega$ in the $\N_*$ norm.
    Then $\lim_{n \rightarrow \infty} Q^*_p(\E_n (\rho) \| \E_n (\eta))$ converges uniformly in $p$ on any compact interval $p \in [1, p_{\max}]$ and for any state $\rho$ satisfying Equation \eqref{eq:ordassump}.
\end{lemma}
\begin{proof}
    Let $\gamma_n := \nu_{\E_n(\rho) | \eta} \eta^{1/2}$ and $\nu_{\E_n(\rho) | \eta}$ be defined such that $\E_n(\rho)^{1/2} = \nu_{\E_n(\rho) | \eta} \eta^{1/2}$ when $\E_n(\rho)^{1/2} \leq C \eta^{1/2}$ for some $C > 0$. Recall that if $\rho^{1/2} \leq C \eta^{1/2}$, then $\E_n(\rho)^{1/2} \leq C \E_n(\eta)^{1/2} = \eta^{1/2}$.
    Recalling that $\E_n(\eta) = \eta$, convergence of the sequence $(\E_n(\rho))_{n=1}^\infty$ implies convergence of the sequence $(\nu_{\E_n(\rho) | \E_n(\eta)} \E_n(\eta)^{1/2})_{n=1}^\infty$. The Kosaki norm $ \| \gamma_n^* \gamma_n \|_{L_{p}^{1/2}(\eta)}$ thereby converges.
    What remains is to show that this convergence can be made uniform for $p \in [1,p_{\max}]$.
    Via Stein's interpolation as in Theorem \ref{thm:stein} and using Remark \ref{rem:sandwich},
    \[ Q^{*}_{p_{\max}/(1 + (p_{\max} - 1) \theta)}(\E_n(\rho) \| \eta) \leq \| \gamma_n^* \gamma_n \|_{L_{p_{\max}}^{1/2}(\eta)}^{1-\theta} \| \gamma_n^* \gamma_n \|_{L_1^{1/2}(\eta)}^\theta \]
    for $\theta \in [0,1]$. Hence $Q^{*}_{p}(\E_n(\rho) \| \eta) \leq \max_\theta \| \gamma_n^* \gamma_n \|_{L_{p_{\max}}^{1/2}(\eta)}^{1-\theta} \| \gamma_n^* \gamma_n \|_{L_1^{1/2}(\eta)}^\theta$. The data processing inequality for sandwiched relative R\'enyi entropy ensures that it eventually increases in $n$. Therefore, $Q^{*}_{p}(\E_n(\rho) \| \eta)$ converges at a rate upper bounded by that of $\max_\theta \| \gamma_n^* \gamma_n \|_{L_{p_{\max}}^{1/2}(\eta)}^{1-\theta} \| \gamma_n^* \gamma_n \|_{L_1^{1/2}(\eta)}^\theta$ for all $p \in [1,p_{\max}]$. This convergence is uniform.
\end{proof}
\begin{lemma} \label{lem:qlim}
Assuming Equation \eqref{eq:ordassump}, 
\[ \lim_{p\to \infty} Q^*_p(\rho\|\eta) \lel Q^*_{\infty}(\rho\|\eta) \pl .\] 
\end{lemma}
\begin{proof} We first show that for  every $\nu$ and $2<p<\infty$
 \[ {\rm Prob}_{\eta}(v>\la)\la^p \kl \|v\eta^{1/2p}\|_p^p \pl .\] 
Indeed, 
 \[ L_p\lel [L_{\infty},L_2]_{2/p} \subset L_{2/p,\infty}\] 
(see  \cite{kosaki_applications_1984}). Then we can find $v=x_1+x_2$ such that
 \[ \|x_1\|_{\infty}+t\|x_2\|_2\kl t^{2/p}\|x\|_p \pl .\]
For $p=2$ Chebyshev's inequality is trivial and hence there exists a projection $e$ with $\|x_2e\|_{\infty}\le \mu$ and 
 \[ \eta(1-e)\mu^2\le \|x_2\|_2^2 \pl .\] 
We choose $\mu=t\|x_2\|_2$ and get
  \[ \|x_2 e\|_{\infty} \kl t^{2/p} \|x\|_p \] 
and 
 \[ \eta(1-e)\le t^{-2} \] 
For a given $\la$ we choose $t$ such that $\la=t^{2/p}\|x\|_p$, i.e  $t^{2/p}\lel \frac{\la}{\|x\|_p}$. This gives $t^{-2}=\frac{\|x\|_p^p}{\la^p}$. 

With the help of the Chebychev inequality it is not hard to conclude. Assume that $\lim_p \|x\|_p\le \gamma < \|x\|_{\infty}$. Let $(1+\eps)\gamma<\|x\|_{\infty}$. Then we can find $e_p$ such that  
 \[ \eta(1-e_p)((1+\eps)\gamma)^p\kl \gamma^p \pl .\] 
 and $xe_p\kl (1+\eps)\gamma$. Since $(1+\eps)^{-p}$ converges to $0$, we deduce that $e_p$ converges to $1$ in the strong operator topology and hence 
 \[ \|x\|_{\infty} \kl \lim_p \|xe_p\|_{\infty}\kl (1+\eps)\gamma \pl .\] 
This contradiction concludes the proof.
\end{proof}
\begin{rem} The scale of $L_p$ norms is continuous. Indeed, since $\| \cdot \|_{L_p^{1/2}}$ is monotone in $p$, we just have to argue that $\|x\|_p\kl \limsup_{q\to p} \|x\|_q =: \gamma$. Recall that 
 \[ \|x\|_p \kl \|x\|_{\infty}^{1-q/p} \|x\|_q^{q/p} \kl
 \|x\|_{\infty}^{1-q/p} \gamma^{q/p}
 \] 
 holds for $\frac{1}{p}=\frac{1-\theta}{\infty}+\frac{\theta}{q}$. Sending $q\to p$ yields $\|x\|_p\le \gamma$.   
\end{rem}    
\begin{lemma} \label{lem:compact1}
    Consider a von Neumann algebra $\N$ and any states $\rho, \eta \in \N_*$ obeying Equation \eqref{eq:ordassump}. Then
    \begin{itemize}
        \item For every $r \leq D_1(\rho \| \eta)$ or if $r = 0$, $H_r^*(\rho \| \eta) = 0$, and the supremum is achieved at $\alpha = 1$. Hence for every $\alpha_0 \in (1, \infty)$,
         \[ 0 = H_r^*(\rho \| \eta) \leq \sup_{\alpha \in (1,  \alpha_0]} \frac{\alpha - 1}{\alpha} \Big (  r - D^*_\alpha(\rho \| \eta) \Big ) \pl. \]
        \item For every $r \in (D^* (\rho \| \eta), D^*_\infty(\rho \| \eta))$, there exists an $\alpha_0 \in (1, \infty)$ for which
        \[ H_r^*(\rho \| \eta) \leq \sup_{\alpha \in (1,  \alpha_0]} \frac{\alpha - 1}{\alpha} \Big (  r - D^*_\alpha(\rho \| \eta) \Big ) \pl. \]
        \item For every $r \geq D^*_\infty(\rho \| \eta) > 0$ and $\epsilon > 0$, there exists an $\alpha_0 \in (1, \infty)$ for which
        \[ H_r^*(\rho \| \eta) \leq \sup_{\alpha \in (1,  \alpha_0]} \Big ( \frac{\alpha - 1}{\alpha} r - \ln Q^*_\alpha(\rho \| \eta) \Big ) + \epsilon \pl. \]    \end{itemize}
\end{lemma}
\begin{proof}
    Recall Equation \eqref{eq:hoeffding1}. Moreover, recall that $D^*_\alpha(\rho \| \eta)$ is non-decreasing in $\alpha \in [1, \infty]$. If $r \leq D^*_1(\rho \| \eta)$, then it is clear that $r - D^*_\alpha(\rho \| \eta) < 0$ for all $\alpha \geq 1$. Also, within this range,
    \[  \frac{\alpha - 1}{\alpha} \big ( r - D_\alpha^*(\rho \| \eta) \big ) \]
    is monotonically decreasing in $\alpha$, because increasing $\alpha$ makes $D_\alpha^*(\rho \| \eta)$ larger by the monotonicy of relative entropy in $\alpha$, and increasing $\alpha$ makes $(\alpha - 1)/\alpha$ larger, multiplying a negative number. Therefore, the supremum is achieved at $\alpha = 1$.
    
    If $r \in (D^* (\rho \| \eta), D^*_\infty(\rho \| \eta))$, there exists an $\alpha_0$ such that for all $1 < \alpha \leq \alpha_0$, $r - D^*_\alpha(\rho \| \eta) > 0$ and for all $\alpha > \alpha_0$, $r - D^*_\alpha(\rho \| \eta) < 0$. Hence the supremum is achieved with $\alpha \in (1, \alpha_0]$.

    If $r \geq D^*_\infty(\rho \| \eta)$, then the supremum in Equation \eqref{eq:hoeffding1} may require arbitrarily large $\alpha$, including $\alpha = \infty$.
    For every $\delta > 0$, there exists an $\alpha_1 < \infty$ sufficiently large that
    \begin{equation} \label{eq:fraccompare1}
        \frac{\alpha_2 - 1}{\alpha_2} \geq (1 - \delta) \geq \frac{\alpha - 1}{\alpha}  - \delta
    \end{equation}
    for every $\alpha > 1$ and every $\alpha_2 > \alpha_1$. Recall that under Equation \eqref{eq:ordassump}, $D_\infty^*(\rho \| \eta) < a$ for the assumed constant $a > 0$. Moreover, recall that $D_\alpha^*(\rho \| \eta)$ is non-decreasing in $\alpha$. 
    Therefore,
    \[ H_r^*(\rho \| \eta) \leq \sup_{\alpha \in (1,\alpha_0)}
        \Big (\frac{\alpha - 1}{\alpha} r - D_\alpha^*(\rho \| \eta) \Big ) + \delta r  \pl. \]
    To complete the $r \geq D^*_\infty(\rho \| \eta)$ case, set $\delta = \epsilon / r$.
\end{proof}

\begin{lemma} \label{lem:haageruphoeffding}
    Let $(\M, \eta)$ repsectively be a von Neumann algebra and normal, faithful state. Let $G = \bigcup_{n} 2^{-n}\zz\subset \rz$ and $\M \rtimes G$ denote the crossed product with respect to $\eta$. Let $\hat{\eta}$ denote the canonical embedding of $\eta$ in $(\M \rtimes G)_*$. Let
    \[ H^{* \leq \alpha_0}_r(\rho \| \eta) := \sup_{\alpha \in (1, \alpha_0]} \frac{\alpha - 1}{\alpha} \big (  r -  D^*_\alpha(\rho \| \eta) \big ) \pl. \]
    Then for every $\alpha_0 \in (1, \infty)$ and $\rho \in \M_*$ obeying Equation \eqref{eq:ordassump}, there exists a sequence of finite von Neumann algebras $(\M_n \subseteq \M \rtimes G)_{n=1}^\infty$ such that
    \[ H^{* \leq \alpha_0}_r(\rho \| \eta) = \lim_{n \rightarrow \infty} H^{*\leq \alpha_0}_r(\E_n(\hat{\rho}) \| \hat{\eta}) \pl, \]
    where $\hat{\rho}$ denotes the canonical embedding of $\rho$ into $(\M \rtimes G)_*$.
\end{lemma}
\begin{proof}
    Recall Theorem \ref{thm:haagerup} and herein use the notation therein. Let $\rho_n := \E_n(\hat{\rho})$, and recall that 
    \[ H^{* \leq \alpha_0}_r(\rho \| \eta) 
        = \sup_{\alpha \in (1, \alpha_0]} \Big ( \frac{\alpha - 1}{\alpha} r - \ln Q^*_\alpha(\rho \| \eta) \Big ) \pl. \]
    By Lemma \ref{lem:haagerupconv}, $\lim_{n \rightarrow \infty} Q_{\alpha}^*(\rho_n \| \eta_n) = Q_{\alpha}^*(\hat{\rho} \| \eta)$ uniformly on $\alpha$ in the compact interval $[1, \alpha_0]$. Equation \eqref{eq:ordassump} implies that $Q_{\alpha}^* \in (0, \infty)$, so the logarithm is continuous. Therefore,
    \[ \lim_{n \rightarrow \infty} H_r^{* \leq \alpha_0}(\E_n(\rho) \| \hat{\eta}) = H_r^{* \leq \alpha_0}(\hat{\rho} \| \hat{\eta}) \pl. \]
    Furthermore, because the embedding from $\M_*$ into $(\M \rtimes G)_*$ is faithful,
    \[ H_r^{* \leq \alpha_0}(\hat{\rho} \| \hat{\eta}) = H_r^{* \leq \alpha_0}(\rho \| \eta) \pl. \]
    
\end{proof}

\subsection{Approximation via Finite Spectrum} \label{sec:approx-finite-spec}
The main idea of this Section is to approximate densities in the crossed product and approximating type $II_1$ sequence by densities with finite spectrum. When densities have finite spectrum, we can apply the method of types \cite{hiai_proper_1991, hayashi_optimal_2002}, one of the foundational tools of information theory.
\begin{rem}
    In a finite von Neumann algebra, there is a canonical, invertible map between a state $\rho$ and its density $d_\rho$ with respect to the finite trace. In this section and thereafter, we interpret expressions that mix states with densities as applying the canonical mapping where needed. For example, we may interpret $D_\alpha^*(\rho \| d_\eta)$ as $D_\alpha^*(d_\rho \| d_\eta)$. It also holds by data processing and the invertibility of this map that $D_\alpha^*(\rho \| \eta) = D_\alpha^*(d_\rho \| d_\eta)$ for all $\rho, \eta$.
\end{rem}
\begin{lemma} \label{lem:finitespec}
Let $\delta< d_{\eta} <\delta^{-1}$ be a bounded density in a finite von Neumann algebra. Then there exists a sequence of conditional expectations $(F_k)_{k=1}^\infty$ such that $(1+1/k)^{-1} d_{\eta} \leq F_k(d_\eta) \leq (1+1/k) d_{\eta}$ for sufficiently large $k$,
\[ |  D_p^*(\rho \| \eta) - D_p^*(d_\rho \| F_k(d_\eta)) | \kl \frac{1}{k} \]
for all $p \in (1, \infty]$, and
 \[ |H_r^*(\rho\|\eta)-H_r^*(d_\rho\|F_k(d_\eta))|\kl \frac{1}{k} \pl .\]
For fixed $k$, $|\text{spec}(F_k(d_\eta))|\kl 2k \delta^{-2}$, where $\text{spec}(X)$ denotes the spectrum of an operator $X$.
\end{lemma}
\begin{proof}
    Recall the spectral decomposition
    \[ d_\eta(\mu) = \int_0^\infty x d \mu(x) \pl. \]
    Let $c > 1$, and define
    \[ d_a^- = \delta \sum_{j=0}^k a^j 1_{[\delta a^j, \delta a^{j+1})}(d_\eta), \pla d^+_a = \delta \sum_{j=0}^k a^{j+1} 1_{[\delta a^j, \delta a^{j+1})}(d_\eta) \]
    for some $k \in \NN$. Thanks to the boundedness assumption, only finitely many of the spectral projections $e_j := 1_{[\delta a^j, \delta a^{j+1})}$ are non-trivial, so a finte number of $(a_j)$ suffice. Then set $a$ to be the smallest for which
    \[ d_a^- \leq d_\eta \leq d_a^+ \leq a d_a^- \pl. \]
     Let $M_a = \text{span}(\{ e_j \} )$ be the commutative, finite subalgebra generated by these projections and $F_a$ be the conditional expectation to that subalgebra. Then $d_a^- \leq F_a(d_\eta) \leq d_a^+$, so $a^{-1} F_a(d_\eta) \leq d_\eta \leq a F_a( d_\eta)$. For any densities $g \geq 0$ and $d > 0$,
    \[ \|d^{-1/2p'} g d^{-1/2p'} \|_p = \|g^{1/2} d^{-1/2p'} \|_{2p}^2 = \|g^{1/2} d^{-1/p'} g^{1/2} \|_p \pl. \]
    Now we note for any $d_1$ and $d_2$ that  $d_1\le Kd_2$ implies  
     \[ \|g^{1/2}d_2^{-1/p'}g^{1/2}\|_p \kl K^{1/p'} \|g^{1/2}d_1^{-1/p'}g^{1/2}\|_p \pl .\] 
    In our situation, we deduce
     \[  \|d_\eta^{-1/2p'} g d_\eta^{-1/2p'}\|_p^p \kl a^{p/p'}
     \|F_a(d_\eta)^{-1/2p'} g F_a(d_\eta)^{-1/2p'}\|_p^p \kl a^{2p/p'} \|d_\eta^{-1/2p'} g d_\eta^{-1/2p'}\|_p^p \pl .\] 
    Taking the logarithm,
     \[ \big |\ln \|d^{-1/2p'}gd^{-1/2p'}\|_p -
     \ln \|F_a(d_\eta)^{-1/2p'}g F_a(d_\eta)^{-1/2p'}\|_p \big | 
     \kl \frac{1}{p'} \ln a \pl .\] 
    Recall that $D_p^*(\rho \| \eta) = p' \ln Q_p^*(\rho \| \eta)$. Therefore,
    \[ \big | D_p^*(\rho \| \eta) -D_p^*(d_\rho \| F_a(d_\eta)) \|_p \big | 
     \kl \ln a \]
    Moreover, for $p=\al \in (1, \infty]$, we see that $\frac{1}{p'}=\frac{\al-1}{\al} \leq 1$. Now, we may choose $a=1+\frac{1}{k}$, so $\ln a \leq \frac{1}{k}$. Setting $g = \rho = \nu \eta^{1/2}$ and $d = d_\phi$ as in Remark \ref{rem:sandwich},
    \[ \big |\ln Q_{\alpha}^*(\rho \| \eta) -
     \ln Q_{\alpha}^*(\rho \| F_{1/k}(\eta)) \big | 
     \kl \frac{1}{k} \pl .\]
     An upper bound for $H_r^*(\rho \| F_k(\eta))$ over $H_r^*(\rho \| \eta)$ follows from assuming that $\ln Q_{\alpha}^*(\rho \| F_k(\eta))$ is $1/k$ smaller than $\ln Q_{\alpha}^*(\rho \| \eta)$ for the value of $\alpha$ already achieving the supremum as in Remark \ref{rem:hoeffdingequiv}. By similar argument, the most negative decrease is $1/k$. Hence $1/k$ bounds the $H_r^*$ difference.

    According to our boundedness assumption, $|\si(F_a(d_\phi))|\le j_2-j_1$ for some $j_1, j_2 \in \NN$ such that $a^{-j_1}\le \delta$ and $a^{j_2}\le \delta^{-1}$. This means
     \[ (j_2-j_1)\ln a \lel \delta^{-2} \pl .\] 
    For our choice, $j_2-j_1\le 2 k\delta^{-2}$.
\end{proof}

The next lemma is a version of the well-known method of types in (quantum) information theory. For many copies of a density with finite spectrum, though the dimension scales exponentially in the number of copies, the number of eigenvalues scales only polynomially \cite{hiai_proper_1991, hayashi_optimal_2002}. We also use the cp-order index, equivalent for finite-dimensional conditional expectations to the Pimsner-Popa index \cite{pimsner_entropy_1986}, given by
\[
    C_{cp}(\E \| \F) = \inf \{ c | c (\F \otimes \text{Id})(\rho) \geq (\E \otimes \text{Id})(\rho) \} \pl,
\]
where the identity is taken in an algebra of the same type as projections $\E$ and $\F$. The Pimsner-Popa index is in turn a finite-dimensional analog of the Jones index \cite{jones_index_1983}.

\begin{lemma} \label{lem:types}
Let $N$ be a finite von Neumann algebra and $d=\sum_{j=1}^K d_j e_j$ a finite density with finite spectrum, where $(e_j)_{j=1}^K$ is a family of projections for which $\sum_{j=1}^K e_j = 1$. Then...
\begin{enumerate}
\item[i)] If $F$ is the conditional expectation onto $\braket{d}$, the algebra of operators that commute with $d$, then $F$ has cp-order index  $\le K$ with respect to the identity.
\item[ii)] The conditional expectation onto $\braket{d^{\ten n}}$ has cp-order index $\le (n+1)^K$. 
\item[iii)]   
 \begin{align*}
   H_{r}(\rho|d) \lel \frac{1}{n} H_{rn}(\rho^{\ten n}|d^{\ten n}) 
 &\leq  \frac{1}{n} H_{rn}(F_n(\rho^{\ten n})|d^{\ten n}) \\
 &\leq \frac{1}{n} H_{rn}(\rho^{\ten n}|d^{\ten n})+ K \frac{\log (n+1)}{n}  \\
 & \lel H_r(\rho|d) + K \frac{\log (n+1)}{n} \pl . 
  \end{align*} 
\end{enumerate}
\end{lemma}      

\begin{proof} The conditional expectation is given by  
 \[ F(x) \lel \sum_j e_jxe_j \pl ,\] 
Let $x\gl 0$, $N\subset \mathbb{B}(H)$, and $h$ be any vector in $H$. We note the subspace decomposition $h = \sum_{j=1}^K e_j h$. Therefore,
  \begin{align*} 
  (h,xh) & \lel \sum_{j,j'} (x^{1/2}e_jh,x^{1/2}e_{j'}h) \\
  &\le \sum_{j,j'} \|x^{1/2}e_j h\|\|x^{1/2}e_{j'} h\| 
  \lel \Big ( \sum_{j=1}^K (h,e_jxe_jh)^{1/2} \Big )^2 \\ 
  &\kl K \sum_j (h,F(x)h) \pl .
  \end{align*}
As the inequality holds for arbitrary $h \in H$, it implies that $K F(x) \geq_{cp} x$, proving the first assertion. For the second assertion, we observe that, thanks to commutativity  
 \[ d^{\ten n} \lel \sum_{j_1+\cdots+ j_K=n} d_1^{j_1}\cdots d_K^{j_K} \sum_{\#{l|i_l=1}=j_1,...\#{l|i_l=K}=j_k}   e_{i_1}\ten \cdots \ten e_{i_n}  \pl.  
 \]   
In $\rz^K$ we consider the simplex $S_{n+1}=\{(x_1,...,x_K)| 0 \leq x_j, \sum_j x_j\le n+1\}$. This contains our discrete points and an $p+\stackrel{\circ}{S}_1$, where $\stackrel{\circ}{S}_1$ denotes the interior, around it. Thus
 \[ \# \{(j_1,...,j_k) |\sum_l j_l=n\} \kl \frac{vol(S_{n+1})}{vol(S_{1})} \lel (n+1)^K \pl .
 \] 
For the last assertion, we fix $K$, $d$ and $\nen$ and the conditional expectation $F$ onto $\{d^{\ten_n}\}'$ and use the notation  $\rho_n=\rho^{\ten_n}$, $d_n=d^{\ten_n}$. We recall that 
 \[ Q_{\al}^*(F(\rho_n) \| F(d_n))\kl Q_{\al}^*(\rho_n \| d_n) \pl . \] 
On the other hand $\rho_n \le {\rm Ind}_n F(\rho_n)$ implies 
 \[ \|d_n^{-1/2p'}\rho_nd_n^{-1/2p'}\|_p \kl {\rm Ind} 
 \|d_n^{-1/2p'}F(\rho_n) d_n^{-1/2p'}\|_p \pl ,\] 
where ${\rm Ind} = (n+1)^K$. This implies
 \[ \ln Q_{\al}^*(F(\rho_n) \| d_n) \kl \ln Q_{\al}^*(\rho_n \| d_n)
 \kl \ln Q_{\al}^*(F(\rho_n) \| d_n)+ K \ln (n+1)  \pl.    \] 
Since $Q_\alpha^*$ is uniformly bounded in $\alpha$, $H_r^*$ is correspondingly bounded in the reverse direction.
\end{proof}

\subsection{Finding Tests on the Approximating States}
Once densities are approximated via operators with finite spectrum, what remains is to show that tests in one setting transfer to or are replaced in another. One such transference is from the approximating finite von Neumann algebras to the crossed product and ultimately the original algebra. Another is between various many-copy limits.
\begin{rem} \label{rem:br-lim-sup}
We will work with
 \begin{align*} \hat{B}_r(\rho\|\eta)&= \inf\{R| \exists_{n_0} 0\le T_n\le 1| \eta^{\ten_n}(T_n)\le e^{-rn} \rho^{\ten_n}(T_n)\gl e^{-Rn} \} \\
 &=  
 \limsup_{n} \inf\{\frac{-\log \rho(T_n)}{n}| \eta^{\ten_n}\pl ,\pl (T_n)\le e^{-rn}\} \pl .
 \end{align*}
Mosoyni-Ogawa's definition as in \cite[Equation (43)]{mosonyi_quantum_2015} and our Equation \eqref{eq:Br} is slightly different 
 \[ B_r(\rho\|\eta) \lel \inf \bigg\{ R | \exists (T_n)  \limsup \frac{\log\eta^{\ten_n}(T_n)}{n}\kl  -r \limsup_n \frac{\log \rho^{\ten_n}(T_n)}{n} \gl -R \bigg \} \pl.\] 
Clearly, our assumption $\eta^{\ten_n}(T_n)\le e^{-rn}$ is stronger than $\limsup_n \frac{\log \eta^{\ten_n}}{n}\le -r$. However, given $\eps>0$ we can find $T_n$ and $n_1$ such that $\eta^{\ten_n}(T_n)\le e^{-(r-\eps)n}$ holds for $n\gl n_1$ as well as and $\rho^{\ten_n}(T_n)\gl e^{-Rn}$.  We may then use $T_n'=e^{-\eps n}T_n$ satisfying the stricter condition and proving us with a $R-\eps$. Since $\eps>0$ is arbitrary, we see that 
 \[ B_r(\rho\|\eta)  \lel \hat{B}_r(\rho\|\eta) \]
and for $\hat{B}_r(\rho\|\eta)$ the limit exists as proven  above.  
\end{rem} 

\begin{prop} \label{prop:Br-alpha}
    If $\rho$ and $\eta$ are states satisfying Equation \eqref{eq:ordassump}, then Equation \eqref{eq:bralpha} stating that 
    \begin{equation*}
        B_r(\rho\|\eta) 
     = - \lim_{n \rightarrow \infty} \frac{1}{n} \log \Big \{ 1 - \alpha_{e^{-nr}}^* (\rho^{\otimes n} \| \eta^{\otimes n}) \Big \}
    \end{equation*}
    holds, and the limit therein exists.
\end{prop}
\begin{proof}
    Recalling the statement of Equation \eqref{eq:bralpha} and its relation to Equation \eqref{eq:Br}, we first aim to show that
    \begin{equation}
        - \lim_{n \rightarrow \infty} \frac{1}{n}
        \log \max_{0 \leq T_n \leq 1} \big \{ \rho^{\otimes n}(T_n) : \eta^{\otimes n}(T_n) \leq e^{- n r} \big \}
    \end{equation}
    exists and equals
    \begin{equation*}
    \begin{split}
    \hat{B}_r(\rho||\eta) := \limsup_n \bigg \{-\frac{\log \rho^{\ten_n}(T_n)}{n}| \eta^{\ten_n}(T^n)\le e^{-rn} \bigg \} \pl.
    \end{split}
    \end{equation*}
    For each fixed $n \in \NN$, a chosen test $T_n$, and $R > 0$,
    \begin{equation}
        \rho^{\otimes n}(T_n) \geq e^{- n R} \pll \equiv \pll R \geq - \frac{1}{n}  \log (\rho^{\otimes n} (T_n)) \pl.
    \end{equation}
    Therefore,
    \begin{equation} \label{eq:rewritings-1}
    \begin{split}
        & - \frac{1}{n} \log \max_{0 \leq T_n \leq 1} \big \{ \rho^{\otimes n}(T_n) : \eta^{\otimes n}(T_n) \leq e^{- n r} \big \}
        \\ & = \min_{0 \leq T_n \leq 1} \Big \{ - \frac{1}{n} \log  \rho^{\otimes n}(T_n) : \eta^{\otimes n}(T_n) \leq e^{- n r} \Big \}
        \\ & = \min \Big \{ R | \exists T_n,  0 \leq T_n \leq 1, R \geq - \frac{1}{n} \log   \rho^{\otimes n}(T_n) , \eta^{\otimes n}(T_n) \leq e^{- n r} \Big \}
        \\ & = \min \big \{ R | \exists T_n,  0 \leq T_n \leq 1, \rho^{\otimes n}(T_n) \geq e^{- n R} , \eta^{\otimes n}(T_n) \leq e^{- n r} \big \} \pl.
    \end{split}
    \end{equation}
    For every $\rho, \eta, r$, by the definition of $\hat{B}_r(\rho \| \eta)$, for every $R \geq \hat{B}_r(\rho \| \eta)$, $\exists (T_n)_{n=1}^\infty$ and $n_0 \in \NN$ such that $\forall n > n_0$,
    \begin{equation*}
        \eta^{\otimes n} (T_n) \leq e^{- r n} \text{ , and } \rho^{\otimes n} (T_n) \geq e^{-R n} \pl.
    \end{equation*}
    Equivalently to the 2nd Equation above, $R \geq - \frac{1}{n}  \log (\rho^{\otimes n} (T_n))$. Therefore, for sufficiently large $n$, there exists a sequence of tests for which
    \begin{equation*}
        R \geq  - \frac{1}{n} \log  \rho^{\otimes n}(T_n) \text{ , and } \eta^{\otimes n}(T_n) \leq e^{- n r} \pl.
    \end{equation*}
    Therefore,
    \begin{equation} \label{eq:br-alpha-1}
        \hat{B}_r(\rho \| \eta) \geq \limsup_{n \rightarrow \infty} \min_{0 \leq T_n \leq 1} \Big \{ - \frac{1}{n} \log  \rho^{\otimes n}(T_n) : \eta^{\otimes n}(T_n) \leq e^{- n r} \Big \} \pl.
    \end{equation}
    For the other direction, assume that
    \begin{equation*}
        R' = \liminf_{n \rightarrow \infty} \min_{0 \leq T_n \leq 1} \Big \{ - \frac{1}{n} \log  \rho^{\otimes n}(T_n) : \eta^{\otimes n}(T_n) \leq e^{- n r} \Big \} \pl.
    \end{equation*}
    Hence $\exists n_0 \in \NN$ such that $\forall n > n_0$,
    \begin{equation*}
        \min_{0 \leq T_n \leq 1} \Big \{ - \frac{1}{n} \log  \rho^{\otimes n}(T_n) : \eta^{\otimes n}(T_n) \leq e^{- n r} \Big \} \geq R' \pl.
    \end{equation*}
    Equivalently, for every sequence of tests $(T_n)_{n=1}^\infty$ such that $\eta^{\otimes n}(T_n) \leq e^{- n r}$,
    \begin{equation*}
         - \frac{1}{n} \log  \rho^{\otimes n}(T_n)  \geq R' \text{ , equivalenlty }
         \rho^{\otimes n}(T_n) \leq e^{- n R'}
    \end{equation*}
    for large enough $n$. However, it also holds that for every $\epsilon > 0$, $\exists$ a sequence of tests $(T_n)_{n=1}^\infty$ such that $\eta^{\otimes n}(T_n) \leq e^{- n r}$,
    \begin{equation*}
        - \frac{1}{n} \log \rho^{\otimes n}(T_n) \leq R' + \epsilon
            \text{ , equivalently } \rho^{\otimes n}(T_n) \geq e^{- n (R' + \epsilon)}
    \end{equation*}
    for infinitely many values of $n$. Let $n_0$ be such a value. Subsequently, we follow a line of argument suggested in the proof of \cite[Theorem 4.4]{mosonyi_quantum_2015}. For every $k \in \NN$,
    \begin{equation*}
        \eta^{\otimes k n_0}(T_{n_0}^{\otimes k}) \leq e^{- k n_0 r} \text{ , and }
        \rho^{\otimes k n_0}(T_{n_0}^{\otimes k}) \geq e^{- k n_0 (R' + \epsilon)} \pl.
    \end{equation*}
    Now consider a sequence of tests $\tilde{T}_n$ given by
    \begin{equation*}
        \tilde{T}_n := e^{- r(n - k n_0)} T_{n_0}^{\otimes k} \otimes \id^{\otimes (n - k n_0)}
    \end{equation*}
    where $k$ is chosen as the largest for which $k n_0 \leq n$. Observe that $\tilde{T}_n(\eta^{\otimes n}) \leq e^{- r(n - k n_0)} e^{- r k n_0} = e^{- r n}$. Furthermore,
    \begin{equation*}
        \rho^{\otimes n}(\tilde{T}_n) \geq e^{- r(n - k n_0)} e^{- k n_0 (R' + \epsilon)} 
            = e^{- k n_0 (R' + \epsilon - (n / k n_0 - 1) r)} \pl.
    \end{equation*}
    Taking $n_0$ fixed for sufficiently large $k$, $n / k n_0$ becomes arbitrarily close to 1. Hence for any $\epsilon', \epsilon'' > 0$,
    \begin{equation*}
        \rho^{\otimes n}(\tilde{T}_n) \geq e^{- k n_0 (R' + \epsilon')} 
        \geq e^{- n (R' + \epsilon'')} 
    \end{equation*}
    for \textit{all} sufficiently large $n$. Hence $(\tilde{T}_n)$ is a sequence of tests for which
    \begin{equation*}
        \limsup_{n \rightarrow \infty} \eta^{\otimes n}(\tilde{T}_n) \leq e^{- r n}
            \text{ , and } \liminf_{n \rightarrow \infty} \rho^{\otimes n}(\tilde{T}_n) \geq e^{- n (R' + \epsilon'')} \pl.
    \end{equation*}
    Therefore, $\hat{B}_r(\rho \| \eta) \leq R' + \epsilon'$, so taking $\epsilon' \rightarrow 0$,
    \begin{equation} \label{eq:br-alpha-2}
        \hat{B}_r(\rho \| \eta) \leq \liminf_{n \rightarrow \infty} \min_{0 \leq T_n \leq 1} \Big \{ - \frac{1}{n} \log  \rho^{\otimes n}(T_n) : \eta^{\otimes n}(T_n) \leq e^{- n r} \Big \} \pl.
    \end{equation}
    Combining Equations \eqref{eq:br-alpha-1} and \eqref{eq:br-alpha-2},
    \begin{equation*}
    \begin{split}
        & \limsup_{n \rightarrow \infty} \min_{0 \leq T_n \leq 1}
            \Big \{ - \frac{1}{n} \log  \rho^{\otimes n}(T_n)
            : \eta^{\otimes n}(T_n) \leq e^{- n r} \Big \}
        \\ \leq &
        \liminf_{n \rightarrow \infty} \min_{0 \leq T_n \leq 1}
            \Big \{ - \frac{1}{n} \log  \rho^{\otimes n}(T_n)
            : \eta^{\otimes n}(T_n) \leq e^{- n r} \Big \} \pl,
    \end{split}
    \end{equation*}
    proving that the limit exists, and via Equation \eqref{eq:rewritings-1} that the original limits considered exist. Since $\hat{B}_r(\rho \| \eta)$ is between these, it is equal to both and hence to the limit. Then using Remark \ref{rem:br-lim-sup}, we obtain the desired equality for Mosonyi \& Ogawa's original $B_r$ as in Equation \eqref{eq:Br}.

\end{proof}

Equivalently to Proposition \ref{prop:Br-alpha},
\begin{equation}
    - \lim_{n \rightarrow \infty} \frac{1}{n} \log \Big \{ 1 - \alpha_{e^{-nr}}^* (\rho^{\otimes n} \| \eta^{\otimes n}) \Big \} \leq R_0
\end{equation}
if and only if $\exists R \leq R_0, (T_n)_{n=1}^\infty, n_0 \in \NN$ such that $\forall n > n_0$, $\eta^{\otimes n}(T_n) \leq e^{-r n}$, and $\rho^{\otimes n}(T_n) \geq e^{- R n}$.
\begin{rem}
    For every pair of states $\rho, \eta$ obeying Equation \eqref{eq:ordassump} and every $n_0 \in \NN$,
    \begin{align*}
        H_{n_0 r}^*(\rho^{\otimes n_0} \| \eta^{\otimes n_0})
            = \sup_{\alpha > 1} \frac{\alpha - 1}{\alpha} \Big( n_0 r - D_\alpha^*(\rho^{\otimes n_0} \| \eta^{\otimes n_0}) \Big)
        = n_0 H_r^*(\rho \| \eta)
    \end{align*}
    by Equation \eqref{eq:Hr} and the additivity of $D_\alpha^*$ on product states.
\end{rem}
\begin{lemma} \label{lem:subseq} Let $n_0\in \nz$ and assume that  
 \[ B_{rn_0}(\rho^{\ten n_0}\|\eta^{\ten n_0}) \geq n_0 B_{rn_0}(\rho^{\ten n_0}\|\eta^{\ten n_0}) .\]    
\end{lemma}
\begin{proof}
    The argument of this proof resembles part of the proof of \ref{prop:Br-alpha} and uses a similar idea inspired by \cite[Theorem 4.4]{mosonyi_quantum_2015}. By Proposition \ref{prop:Br-alpha},
    \[ B_{n_0 r} (\rho^{\otimes n_0} \| \eta^{\otimes n_0}) 
        = - \lim_{k
        \rightarrow \infty} \frac{1}{k} \log \Big ( \max_{0 \leq T_k \leq 1} \{ \rho^{\otimes k n_0} (T_k) : \eta^{\otimes k n_0}(T_k) \leq e^{- k n_0 r} \} \Big ) \pl. \]
    For any given $\epsilon > 0$, assume that $(T_k)$ is an infinite sequence of tests achieving
    $\rho^{\otimes k n_0} (T_k) \geq e^{- n_0 (R + \epsilon) k}$ for sufficiently large $n$, where $n_0 R = B_{n_0 r} (\rho^{\otimes n_0} \| \eta^{\otimes n_0})$,
    while $\eta^{\otimes k n_0}(T_k) \leq e^{- k n_0 r}$
    for all sufficiently large $k$. For each $k \in \NN$ and $n \in \NN$ for which $k n_0 \leq n < (k+1) n_0$, let $\tilde{T}_{n} := e^{- r (n - k n_0)} T_k \otimes \id^{\otimes (n - k n_0)}$, where $T_k$ is a test in the aforementioned achieving family on $k n_0$ copies. For every $m \in \NN$, $\rho^{\otimes m}(\id) = \eta^{\otimes m}(\id) = 1$. Therefore, for sufficiently large $n$,
    \[ \eta^{\otimes n}(\tilde{T}_{n}) = \eta^{\otimes k n_0 + (n - k n_0)}(\tilde{T}_{n})
    = \eta^{\otimes k n_0}(T_{k}) e^{- (n - k n_0) r} \leq e^{- r n} \pl. \]
    Furthermore,
    \begin{equation*}
        \rho^{\otimes n}(\tilde{T}_n) = \rho^{\otimes k n_0 + (n - k n_0)}(\tilde{T}_{n}) \geq e^{- r(n - k n_0)} e^{- k n_0 (R + \epsilon)} 
            = e^{- k n_0 (R + \epsilon - (n / k n_0 - 1) r)} \pl.
    \end{equation*}
    Taking $n_0$ fixed for sufficiently large $k$, $n / k n_0$ becomes arbitrarily close to 1. Hence for any $\epsilon', \epsilon'' > 0$,
    \begin{equation*}
        \rho^{\otimes n}(\tilde{T}_n) \geq e^{- n_0 (R + \epsilon') k} 
        \geq e^{- n (R + \epsilon'')} 
    \end{equation*}
    for \textit{all} sufficiently large $n$. Hence $(\tilde{T}_n)$ is a sequence of tests for which
    \begin{equation*}
        \limsup_{n \rightarrow \infty} \eta^{\otimes n}(\tilde{T}_n) \leq e^{- r n}
            \text{ , and } \liminf_{n \rightarrow \infty} \rho^{\otimes n}(\tilde{T}_n) \geq e^{- n (R + \epsilon'')} \pl.
    \end{equation*}
    Taking $\epsilon''$ to zero completes the Lemma.
\end{proof}

Results in Subsections \ref{sec:approx-finite-alg} and \ref{sec:approx-finite-spec} primarily showed that one may replace a pair of given states $\rho$ and $\eta$ on an arbitrary von Neumann algebra by better behaved approximations, and value of $H_r^*(\rho \| \eta)$ changes in a controlled way. In principle, this is only half of what we need - we also need to know that $B_r(\rho \| \eta)$ is similarly well-behaved. Several of these approximating densities are however given by applying successive conditional expectations starting from the original densities in the crossed product. Therefore, an extremely useful property is as follows:
\begin{lemma} \label{lem:bdp}
    Let $\Phi_* : \M_* \rightarrow \N_*$ be the dual of any unital, completely positive map $\Phi : \N \rightarrow \M$. Then for any states $\rho$ and $\eta$ on $\M_*$ or corresponding densities,
    \begin{equation}
        B_r(\Phi_*(\rho) \| \Phi_*(\eta)) \geq B(\rho \| \eta) \pl.
    \end{equation}
\end{lemma}
\begin{proof}
    Recall in as in Equation \eqref{eq:Br} that $B_r(\rho \| \eta)$ is the infimum over $R$ for which there exists a sequence of tests satisfying some conditions. If $B_r(\Phi_*(\rho) \| \Phi_*(\eta)) = R_0$, then there exists a sequence of tests $(T_n)_{n=1}^\infty$ satisfying the noted conditions for $\Phi_*(\rho)^{\otimes n}(T_n)$ and $\Phi_*(\eta)^{\otimes n}(T_n)$ in the large $n$ limits. If we take a new family of tests given by $\tilde{T}_n := \Phi(T_n)$, then $\Phi_*(\rho)^{\otimes n}(T_n) = \rho^{\otimes n}(\Phi(T_n))$, and $\Phi_*(\eta)^{\otimes n}(T_n) = \eta^{\otimes n}(\Phi(T_n))$. Therefore, there exists a sequence of tests achieving $R_0$ in the infimum within $B(\rho \| \eta)$.
\end{proof}
The broader intuition for Lemma \ref{lem:bdp} is that $B_r$ satisfies a sort of reverse data processing inequality. Since it is a converse rate, degrading of the input states increases the value of $B_r$. Another useful Lemma handles prodessing of one argument:
\begin{lemma} \label{lem:one-replace}
    Assume that $\eta \leq e^s \tilde{\eta}$ for some states $\eta, \tilde{\eta}$. Then for every state $\rho$ satisfying Equation \eqref{eq:ordassump},
    \begin{equation*}
        B_r(\rho \| \eta) \leq B_r(\rho \| \tilde{\eta}) + s
        \text{ , and }
        B_{r-s}(\rho \| \eta) \leq e^s B_r(\rho \| \tilde{\eta}) \pl.
    \end{equation*}
\end{lemma}
\begin{proof}
    Recall Equation \eqref{eq:Br},
    \[ B_r(\rho \| \tilde{\eta}) := \inf \Big \{ R | \exists \{T_n\}_{n=1}^\infty, 0 \leq T_n \leq 1
    | \limsup_{n \rightarrow \infty} \tilde{\eta}^{\otimes n} (T_n) \leq e^{- r n}, 
      \liminf_{n \rightarrow \infty} \rho^{\otimes n} (T_n) \geq e^{-R n} \} \pl . \]
    Let $(T_n)_{n=1}^\infty$ be a family of tests achieving 
    \[ 
     \limsup_{n \rightarrow \infty} \tilde{\eta}^{\otimes n} (T_n) \leq e^{- r n} \text{ , and }
      \liminf_{n \rightarrow \infty} \rho^{\otimes n} (T_n) \geq e^{-(R + \epsilon) n}
    \]
    for $\epsilon > 0$, which can be made arbitrarily small. Since $\eta \leq e^s \tilde{\eta}$,
    \[ 
        \limsup_{n \rightarrow \infty} \eta^{\otimes n} (T_n)
        \leq \limsup_{n \rightarrow \infty} e^{n s} \tilde{\eta}^{\otimes n} (T_n)
        \leq e^{- (r-s) n} \pl.
    \]
    We therefore define $\tilde{T}_n := e^{- s n} T_n$ for each $n$, which achieves
    \begin{equation} \label{eq:s-bound-1}
        \limsup_{n \rightarrow \infty} \eta^{\otimes n} (\tilde{T}_n)
        \leq e^{- (r-s) n} \pl.
    \end{equation}
    Moreover,
    \[
        \liminf_{n \rightarrow \infty} \rho^{\otimes n} (\tilde{T}_n)
        \geq e^{-(R + \epsilon + s) n} \pl.
    \]
    Since we may take $\epsilon$ arbitrarily small, we find for every $\epsilon > 0$ and achieving sequence of tests such that
    \[
        B_r(\rho \| \tilde{\eta}) \leq B_r(\rho \| \tilde{\eta}) + s \pl.
    \]
    Alternatively, returning to Equation \eqref{eq:s-bound-1}, we could use the original sequence of tests $(T_n)$ to achieve
    \[
        B_{r - s} (\rho \| \tilde{\eta}) \leq B_r(\rho \| \tilde{\eta}) \pl.
    \]
\end{proof}
It is also worth noting the ``easy'' direction of Theorem \ref{thm:main}:
\begin{lemma} \label{lem:converse}
    For any states $\rho, \eta \in \M_*$ on von Neumann algebra $\M$ and any $n \in \NN$,
    \[ B_r(\rho \| \eta) 
    \geq H_r^*(\rho \| \eta) \pl. \]
\end{lemma}
\begin{proof}
    Within finite dimension, this result essentially follows the proof of Mosonyi and Ogawa's \cite[Lemma 4.7]{mosonyi_quantum_2015}. To adapt that proof, we first note that expressions of the form $\tr(\rho^{\otimes n} T_n)$ can be trivially replaced by $\rho^{\otimes n}(T_n)$. Second, and more substantially, their proof uses \cite[Lemma 3.3]{mosonyi_quantum_2015}, which shows monotonicty of sandwiched R\'enyi relative entropy under measurements. It is now well-known that R\'enyi relative entropy obeys the data processing inequality \cite{jencova_renyi_2018}, subsuming monotonicity under measurement. Hence the result holds unmodified in the general von Neumann algebra setting.
\end{proof}
The final step toward Theorem \ref{thm:main} is to show the ``hard'' direction.
\begin{lemma} \label{lem:testseq}
    Let $\rho, \eta$ be states on $\M_*$ satisfying Equation \eqref{eq:ordassump}. Then
    \[ B_r(\rho \| \eta) \leq H_r^*(\rho \| \eta) \pl. \]
\end{lemma}
\begin{proof}
    By Lemma \ref{lem:compact1}, for any $r \in \RR$, and states $\rho, \eta \in \M_*$, and any $\epsilon_1 > 0$, there exists an $\alpha_0$ sufficiently large that
    \begin{equation} \label{eq:compactuse}
        |H^{* \leq \alpha_0}_r(\rho \| \eta) - H^{*}_r(\rho \| \eta) | \leq \epsilon_1 \pl.
    \end{equation}
    By Lemma \ref{lem:haageruphoeffding}, there exists a sequence of finite von Neumann algebras $(\M_k \subseteq \M \rtimes G)_{k=1}^\infty$ with respective restriction maps $(\E_k)$ such that for any $\epsilon_1$, there exists some $k_0$ that for all $k \geq k_0$,
    \begin{equation*} 
        |H_r^{* \leq \alpha_0}(\rho \| \eta) - H_r^{* \leq \alpha_0}(\E_k(\iota(\rho)) \| \iota(\eta))| \leq \epsilon_2 \pl.
    \end{equation*}
    Combining with Equation \eqref{eq:compactuse},
    \begin{equation*} 
        |H_r^{*}(\rho \| \eta) - H_r^{*}(\E_k(\iota(\rho)) \| \iota(\eta))| \leq \epsilon_2 + 2 \epsilon_1 
    \end{equation*}
    for sufficiently large $k$ and $\alpha_0$. Clearly $\E_k(\iota(\rho)) \in (\M_k)_*$, and since $\iota(\eta)$ is invariant under $\E_k$, $\iota(\eta) = \E_k(\iota(\eta)) \in (\M_k)_*$. Let $d_\rho^{(k)}$ denote the density of $\E_k(\iota(\rho))$ with respect to the trace $\tau_k$ on the $k$th von Neumann algebra in sequence. The Haagerup reduction, Theorem \ref{thm:haagerup}, also yields for each $k$ a density $d_\eta^{(k)}$ and $\delta_k \in \RR^+$ for which $\delta_k \leq d_\eta^{(k)} \leq \delta_k^{-1}$, for which
    \begin{equation*}
        |H_r^{*}(\rho \| \eta) - H_r^{*}(d_\rho^{(k)} \| d_\eta^{(k)})| \leq \epsilon_2 + 2 \epsilon_1
    \end{equation*}
    By Lemma \ref{lem:finitespec}, there exists a sequence of conditional expectations $(F_{k,l})_{l=1}^\infty$ such that for each $l$,
    \begin{equation*} 
        |H_r^*(d_{\rho}^{(k)} \| d_\eta^{(k)})-H_r^*(d_{\rho}^{(k)} \|F_{k,l}(d_\eta^{(k)}))|\kl \frac{1}{l} \pl ,
    \end{equation*}
    and each $|\text{spec}(F_{k, l}(d_\eta^{(k)}))|\kl 2 l \delta_k^{-2}$. Let $d_{\eta, k, l} := F_{k,l}(d_\eta^{(k)})$. Let $\tilde{F}_{k,l,n}$ be the dual of the conditional expectation to the commutant algebra of $d_{\eta, k, l}^{\otimes n}$
    By Lemma \ref{lem:types}, for every $n \in \NN$,
    \begin{equation*} 
        H_r^*(d_{\rho}^{(k)}\| d_{\eta, k, l}) \leq \frac{1}{n} H_{r n}^*(\tilde{F}_{k,l,n}(d_{\rho}^{(k) \otimes n}) \| d_{\eta, k, l}^{\otimes n})
        \leq H_r^*(d_\rho^{(k)} \| d_{\eta, k, l}) + K \frac{\log(n+1)}{n} \pl.
    \end{equation*}
    Therefore, for any $\epsilon > 0$, choosing sufficiently large $n$, $l$, $k$, and $\alpha_0$,
    \begin{equation} \label{eq:H-approx}
        \Big | \frac{1}{n} H_{r n}^*(\tilde{F}_{k,l,n}(d_{\rho}^{(k) \otimes n}) \| d_{\eta, k, l}^{\otimes n}) - H^{*}_r(\rho \| \eta) \Big | \leq \epsilon \pl.
    \end{equation}
    
    Because $\tilde{F}_{k,l,n}(d_{\rho}^{(k) \otimes n})$ and $d_{\eta, k, l}^{\otimes n}$ commute, they are contained within the predual of a commuting algebra. That algebra is automatically hyperfinite. Therefore, Hiai and Mosonyi's Theorem \ref{thm:origres} applies, yielding for each $k, l$, $\delta$, and $n$ that
    Via Proposition \ref{prop:Br-alpha},
    \begin{equation*}
        B_r(\tilde{F}_{k,l,n}(d_{\rho}^{(k) \otimes n}) \| d_{\eta, k, l}^{\otimes n}) )
        = H_{r n}^*(\tilde{F}_{k,l,n}(d_{\rho}^{(k) \otimes n}) \| d_{\eta, k, l}^{\otimes n}) ) \pl.
    \end{equation*}
    The definition of $B_r$ as in Equation \eqref{eq:Br} then implies that there exists a sequence of tests $(T_{(k,l,n),m})_{m=1}^\infty$ such that for each $m$, $d_{\eta, k, l}^{\otimes n}(T_{(k,l,n),m}) \leq e^{- m n r}$, and
    \begin{equation*} 
        - \lim_{m \rightarrow \infty} \tilde{F}_{k,l,n}(d_{\rho}^{(k) \otimes n})^m(T_{(k,l,n),m}) = H_{r n}^*(\tilde{F}_{k,l,n}(d_{\rho}^{(k) \otimes n}) \| d_{\eta, k, l}^{\otimes n}) \pl.
    \end{equation*}
    Since $\tilde{F}_{k,l,n}(d_{\rho}^{(k) \otimes n}(\rho)$ applies a conditional expectation under which $\eta$ is invariant,
    \begin{equation}
        B_r(\tilde{F}_{k,l,n}(d_{\rho}^{(k) \otimes n}) \| d_{\eta, k, l}^{\otimes n}))
            \geq B_r(d_{\rho}^{(k) \otimes n} \| d_{\eta, k, l}^{\otimes n})) \pl.
    \end{equation}
    Then by Lemma \ref{lem:subseq},
    \begin{equation}
        B_r(d_{\rho}^{(k)} \| d_{\eta, k, l})) 
            \leq  H_{r n}^*(\tilde{F}_{k,l,n}(d_{\rho}^{(k)}) \| d_{\eta, k, l}) ) \pl.
    \end{equation}
    The next step should undo the replacement $d_\eta^{(k)} \rightarrow d_{\eta, k, l} = F_{k,l}(d_\eta)$ on the ``$B_r$'' side. This follows from recalling that by Lemma \ref{lem:finitespec}, $d_\eta^{(k)} \leq (1+1/k) F_{k,l}(d_\eta^{(k)}) \leq e^{1/k} F_{k,l}(d_\eta^{(k)})$. Then by Lemma \ref{lem:one-replace},
    \begin{equation*}
        B_r(d_{\rho}^{(k)} \| d_{\eta}^{(k)})) \leq B_r(d_{\rho}^{(k)} \| d_{\eta, k, l})) + \frac{1}{k} \pl.
    \end{equation*}
    We may then replace $d_\eta^{(k)}$ be $d_\eta$ and, noting that $d_\eta = \E_k(d_\eta)$, again apply Lemma \ref{lem:bdp} to return from the $k$th finite algebra to the crossed product. Lemma \ref{lem:bdp} applies a final time to the map from the original state space into that of the crossed product. Ultimately,
    \begin{equation*}
        B_r(\rho \| \eta) \leq B_r(\tilde{F}_{k,l,n}(d_{\rho}^{(k) \otimes n}) \| d_{\eta, k, l}^{\otimes n}) ) + \frac{1}{k} \pl,
    \end{equation*}
    and with Equation \eqref{eq:H-approx},
    \begin{equation}
        B_r(\rho \| \eta) \leq  H^{*}_r(\rho \| \eta) + \epsilon + \frac{1}{k} \pl,
    \end{equation}
    where $k$ can be made arbitrarily large and $\epsilon$ arbitrarily small.
\end{proof}
\begin{proof}[Proof of Theorem \ref{thm:main}]
    The Theorem follows from combining Lemma \ref{lem:testseq} with Lemma \ref{lem:converse}.
\end{proof}

\subsection{Cutoff Rates}
As noted in \cite[Section 4.2]{mosonyi_strong_2023}, the Hoeffding anti-divergence gives operational interpretation to an optimized quantity over sandwiched $\alpha$-R\'enyi divergences, not necessarily to the divergences themselves. Hence to complete this interpretation, earlier works \cite{mosonyi_quantum_2015, mosonyi_strong_2023, hiai_quantum_2023} recall the generalized $\kappa$-cutoff rate
\begin{equation} \label{eq:cutoff}
    C_{\kappa}(\rho \| \eta) = \inf \{ r_0 | B_r(\rho \| \eta) \geq \kappa (r - r_0) \forall r > 0 \}
\end{equation}
for each $\kappa > 0$. This cutoff rate was originally introduced in \cite{csiszar_generalized_1995}. We use Proposition \ref{prop:Br-alpha} to match out Equation \eqref{eq:cutoff} with \cite[Equation (95)]{mosonyi_quantum_2015}. We hereby show that \cite[Theorem 4.18]{mosonyi_quantum_2015} generalizes to the von Neumann algebra settting:
\begin{cor} \label{cor:cutoff}
    Assume Equation \eqref{eq:ordassump} for states $\rho$ and $\eta$ on a von Neumann algebra. Then
    \begin{equation}
        D^*_\alpha(\rho \| \eta) = C_{(\alpha - 1)/\alpha}(\rho \| \eta) \pl.
    \end{equation}
\end{cor}
\begin{proof}
    Using Theorem \ref{thm:main}, we obtain from Equation \eqref{eq:cutoff} that
    \begin{equation}
        C_{\kappa}(\rho \| \eta) = \inf \{ r_0 | H_r^*(\rho \| \eta) \geq \kappa (r - r_0) \pl \forall r > 0 \} \pl.
    \end{equation}
    Recalling Equation \eqref{eq:Hr}
    \begin{equation*}
    H_r^*(\rho \| \eta) = \sup_{\alpha > 1} \frac{\alpha - 1}{\alpha}
        \Big \{ r - D_\alpha^*(\rho \| \eta) \Big \}  \pl.
    \end{equation*}
    Let $\kappa = (\alpha - 1) / \alpha$. Then
    \begin{equation*}
        H_r^*(\rho \| \eta) = \sup_{0 \leq \kappa < 1} \kappa
        \Big \{ r - D_{1/(1-\kappa)}^*(\rho \| \eta) \Big \}  \pl.
    \end{equation*}
    Observe that therefore, $D_{1/(1-\kappa)}^*(\rho \| \eta)$ is the minimum value of $r_0$ for which
    \[ 
        H_r^*(\rho \| \eta) \geq \kappa (r - r_0) \pl \forall r > 0 \pl.
    \]
\end{proof}
\section{Discussion and Outlook}
We expect that the methods of this work generalize to other scenarios in quantum information theory. The following Remark summarizes the main requirements for such a proof:
\begin{rem} \normalfont The key argument in the original proof for connection between R\'enyi entropies and hypothesis testing is the reduction to the commutative scenario where large deviation applies. This has not changed in the type III situation. The application of Haagerup reduction method may apply to comparing other quantities $B_{1,2}(\rho,\eta)$ and under the following circumstances. 
 \begin{enumerate}
     \item[i)] Both quantities are increasing under channels \item[ii)] $B_j(\rho^{\otimes n},\eta^{\otimes n})=nB_j(\rho|\eta)$.
     \item[iii)] $M=\bigcup_k \M_k$ admits conditional expectations $(\E_k)$ such that
        \[ \lim_k B_j(\E_k(\rho),\E_k(\eta))=B_j(\rho|\eta) \pl. \]
     \item[iv)] Both quantities behave continuously under order perturbations.
     \item[v)]  Both quantities are comparable on commuting densities.
\end{enumerate}
\end{rem}
Ultimately, the Haagerup reduction may transfer bounds from finite to properly infinite von Neumann algebras on quantities that are in some sense well-behaved under conditional expectations and order perturbations. Many such quantities abound in quantum information.

\section{Acknowledgments}
The authors acknowledge conversations with Mil\'an Mosonyi as motivating this research.

 \bibliographystyle{unsrt}
\bibliography{refs}

\begin{thebibliography}{10}

\bibitem{chernoff_measure_1952}
Herman Chernoff.
\newblock A {Measure} of {Asymptotic} {Efficiency} for {Tests} of a {Hypothesis} {Based} on the sum of {Observations}.
\newblock {\em The Annals of Mathematical Statistics}, 23(4):493--507, December 1952.

\bibitem{blahut_hypothesis_1974}
R.~Blahut.
\newblock Hypothesis testing and information theory.
\newblock {\em IEEE Transactions on Information Theory}, 20(4):405--417, July 1974.

\bibitem{han_strong_1989}
T.S. Han and K.~Kobayashi.
\newblock The strong converse theorem for hypothesis testing.
\newblock {\em IEEE Transactions on Information Theory}, 35(1):178--180, January 1989.

\bibitem{csiszar_generalized_1995}
I.~Csiszar.
\newblock Generalized cutoff rates and {Renyi}'s information measures.
\newblock {\em IEEE Transactions on Information Theory}, 41(1):26--34, January 1995.

\bibitem{hiai_proper_1991}
Fumio Hiai and Dénes Petz.
\newblock The proper formula for relative entropy and its asymptotics in quantum probability.
\newblock {\em Communications in Mathematical Physics}, 143(1):99--114, January 1991.

\bibitem{hayashi_optimal_2002}
Masahito Hayashi.
\newblock Optimal sequence of quantum measurements in the sense of {Stein}'s lemma in quantum hypothesis testing.
\newblock {\em Journal of Physics A: Mathematical and General}, 35(50):10759, December 2002.

\bibitem{hayashi_error_2007}
Masahito Hayashi.
\newblock Error exponent in asymmetric quantum hypothesis testing and its application to classical-quantum channel coding.
\newblock {\em Physical Review A}, 76(6):062301, December 2007.

\bibitem{audenaert_discriminating_2007}
K.~M.~R. Audenaert, J.~Calsamiglia, R.~Muñoz-Tapia, E.~Bagan, Ll. Masanes, A.~Acin, and F.~Verstraete.
\newblock Discriminating {States}: {The} {Quantum} {Chernoff} {Bound}.
\newblock {\em Physical Review Letters}, 98(16):160501, April 2007.

\bibitem{jaksic_quantum_2012}
V.~Jakšić, Y.~Ogata, C.-A. Pillet, and R.~Seiringer.
\newblock {Quantum} {Hypothesis} {Testing} {and} {Non}-{equilibrium} {Statistical} {Mechanics}.
\newblock {\em Reviews in Mathematical Physics}, 24(06):1230002, July 2012.

\bibitem{buscemi_information-theoretic_2019}
Francesco Buscemi, David Sutter, and Marco Tomamichel.
\newblock An information-theoretic treatment of quantum dichotomies.
\newblock {\em Quantum}, 3:209, December 2019.

\bibitem{hiai_quantum_2023}
Fumio Hiai and Milán Mosonyi.
\newblock Quantum {R}{\textbackslash}'enyi divergences and the strong converse exponent of state discrimination in operator algebras, January 2023.
\newblock arXiv:2110.07320 [math-ph, physics:quant-ph].

\bibitem{mosonyi_quantum_2015}
Milán Mosonyi and Tomohiro Ogawa.
\newblock Quantum {Hypothesis} {Testing} and the {Operational} {Interpretation} of the {Quantum} {Rényi} {Relative} {Entropies}.
\newblock {\em Communications in Mathematical Physics}, 334(3):1617--1648, March 2015.

\bibitem{wilde_strong_2014}
Mark~M. Wilde, Andreas Winter, and Dong Yang.
\newblock Strong {Converse} for the {Classical} {Capacity} of {Entanglement}-{Breaking} and {Hadamard} {Channels} via a {Sandwiched} {Rényi} {Relative} {Entropy}.
\newblock {\em Communications in Mathematical Physics}, 331(2):593--622, October 2014.

\bibitem{muller-lennert_quantum_2013}
Martin Müller-Lennert, Frédéric Dupuis, Oleg Szehr, Serge Fehr, and Marco Tomamichel.
\newblock On quantum {Rényi} entropies: {A} new generalization and some properties.
\newblock {\em Journal of Mathematical Physics}, 54(12):122203, December 2013.

\bibitem{kosaki_applications_1984}
Hideki Kosaki.
\newblock Applications of the complex interpolation method to a von {N}eumann algebra: noncommutative {$L^{p}$}-spaces.
\newblock {\em J. Funct. Anal.}, 56(1):29--78, 1984.

\bibitem{gao_capacity_2018}
Li~Gao, Marius Junge, and Nicholas LaRacuente.
\newblock Capacity {Estimates} via {Comparison} with {TRO} {Channels}.
\newblock {\em Communications in Mathematical Physics}, 364(1):83--121, November 2018.

\bibitem{mosonyi_strong_2023}
Milán Mosonyi.
\newblock The {Strong} {Converse} {Exponent} of {Discriminating} {Infinite}-{Dimensional} {Quantum} {States}.
\newblock {\em Communications in Mathematical Physics}, 400(1):83--132, May 2023.

\bibitem{cheng_strong_2024}
Hao-Chung Cheng and Li~Gao.
\newblock On {Strong} {Converse} {Theorems} for {Quantum} {Hypothesis} {Testing} and {Channel} {Coding}, March 2024.
\newblock arXiv:2403.13584 version: 1.

\bibitem{witten_aps_2018}
Edward Witten.
\newblock {APS} {Medal} for {Exceptional} {Achievement} in {Research}: {Invited} article on entanglement properties of quantum field theory.
\newblock {\em Reviews of Modern Physics}, 90(4):045003, October 2018.

\bibitem{ji_mip*re_2022}
Zhengfeng Ji, Anand Natarajan, Thomas Vidick, John Wright, and Henry Yuen.
\newblock {MIP}*={RE}, November 2022.
\newblock arXiv:2001.04383 [quant-ph].

\bibitem{araki_classification_1968}
Huzihiro Araki and E.~J. Woods.
\newblock A {Classification} of {Factors}.
\newblock {\em Publications of the Research Institute for Mathematical Sciences}, 4(1):51--130, April 1968.

\bibitem{haagerup_reduction_2010}
Uffe Haagerup, Marius Junge, and Quanhua Xu.
\newblock A reduction method for noncommutative {$L_p$}-spaces and applications.
\newblock {\em Trans. Amer. Math. Soc.}, 362(4):2125--2165, 2010.

\bibitem{junge_multivariate_2022}
Marius Junge and Nicholas LaRacuente.
\newblock Multivariate trace inequalities, p-fidelity, and universal recovery beyond tracial settings.
\newblock {\em Journal of Mathematical Physics}, 63(12):122204, December 2022.

\bibitem{furuya_monotonic_2023}
Keiichiro Furuya, Nima Lashkari, and Shoy Ouseph.
\newblock Monotonic multi-state quantum \textit{f} -divergences.
\newblock {\em Journal of Mathematical Physics}, 64(4):042203, April 2023.

\bibitem{jencova_renyi_2018}
Anna Jenčová.
\newblock Rényi {Relative} {Entropies} and {Noncommutative} \$\${L}\_p\$\$ {L} p -{Spaces}.
\newblock {\em Annales Henri Poincaré}, 19(8):2513--2542, August 2018.

\bibitem{bergh_interpolation_1976}
J.~Bergh and Jorgen Lofstrom.
\newblock {\em Interpolation Spaces: An Introduction}.
\newblock Grundlehren der mathematischen Wissenschaften. Springer-Verlag, 1976.

\bibitem{pimsner_entropy_1986}
Mihai Pimsner and Sorin Popa.
\newblock Entropy and index for subfactors.
\newblock {\em Annales scientifiques de l'École Normale Supérieure}, 19(1):57--106, 1986.

\bibitem{jones_index_1983}
V.~F.~R. Jones.
\newblock Index for subfactors.
\newblock {\em Inventiones mathematicae}, 72(1):1--25, February 1983.

\end{thebibliography}

\end{document}